\documentclass{article}
\usepackage{amssymb,amsthm,amsmath}
\usepackage{amsfonts}
\usepackage{graphics}
\usepackage[latin1]{inputenc}
\usepackage{color}
\usepackage{float}
\usepackage{lmodern} 
\usepackage[left=1.5in,right=1.3in,top=1.1in,bottom=1.1in,includefoot,includehead,headheight=13.6pt]{geometry}

\usepackage[dvips]{graphicx,epsfig}
\usepackage{subfigure}
\usepackage{latexsym,euscript,epic,eepic}
\usepackage{float,caption}
\usepackage{subfigure}
\usepackage{bbding}
\newtheorem{Theo}{Theorem}

\newtheorem{Lem}{Lemma}

\begin{document}
\title{{\sc Adaptive Realized Hyperbolic GARCH Process: Stability and Estimation
}}
\author{El Hadji Mamadou Sall$^{(1,*)}$,
El Hadji Deme$^{(1)}$ and
Abdou K\^a Diongue $^{(1)}$\\
~~\\
\small $^{(1)}$ \textsl{LERSTAD, UFR SAT Universit\'e Gaston Berger, BP 234 Saint-Louis, S\'en\'egal}~~~~~~~~~~~~~~~~~~~~~~~~~~~~\\
$^{(*)}$ Corresponding autor: \textit{elhadjimamadou\_sall@yahoo.fr}
~~~~~~~~~~~~~~~~~~~~~~~~~~~~~\\
}
 \date{
 \;
 }
\maketitle
\noindent\rule[2pt]{\textwidth}{0.5pt}
\textbf{\large{Abstract.}} ~~\\
In this paper, we propose an Adaptive Realized Hyperbolic GARCH (A-Realized HYGARCH) process to model the long memory of high-frequency time series with possible structural breaks. The structural change is modeled by allowing the intercept to follow the smooth and flexible function form introduced by Gallant in \cite{Gallant1984}. In addition, stability conditions of the process are investigated. 
A Monte Carlo study is investigated in order to illustrate the performance of the A-Realized HYGARCH process
compared to the Realized HYGARCH with or without structural change.\\
\noindent\rule[2pt]{\textwidth}{0.5pt}\\
\noindent \textbf{\large{Keywords and expressions}}: Realized HYGARCH model, high-frequency data, long memory, realized measures, structural changes 
\section{Introduction}
\label{sec1}
Volatility forecast of asset returns is very important for option pricing as well as risk management. The Autoregressive Conditional Heteroskedasticity (ARCH) models introduced by Engle \cite{engle1982} and generalized by Bollerslev \cite{bollerslev1986generalize} are widely used to study the properties of volatility for economic and financial data. 
However, there are several shortcomings with using the Generalized ARCH (GARCH) models for risk management or forecasting volatility. The major issue is the persistence of variance that evolves over time and that the GARCH model cannot handle
. It has been extensively observed and studied in various fields of economic and finance over the last decades
 \cite{aggarwaletal1999, Fanetal2008, Grangerandhyang2004}. \\
 
\noindent The long memory exists in the studies of the volatility of high-frequency for financial time series \cite{BollerslevandBaillie1996, Dacorognaetal1993, Dingetal1993, Grangerandding1996b}. 
A widely accepted definition of long memory is $Var(S_T)=O(T^{2d+1})$, where $S_T=\sum_{t=1}^T y_t$, ${y_t}$ is a sequence of financial series and T is the number of observations, $"d"$ is the long-memory parameter \cite{DieboldandInoue2001}. 
To overcome this problem, many models are introduced in the literature.   Among others, we can cite the  Fractionally Integrated GARCH (FIGARCH) model proposed by  Bollerslev and Baillie \cite{BollerslevandBaillie1996}, the Saisonal-HYGARCH model by Diongue and Guegan \cite{diongue2007stationary}, the New HYGARCH by Li \textit{et al.} \cite{Lietal2014} and the Hyperbolic GARCH (HYGARCH) proposed by Davidson \cite{Davidson(2004)}.\\

\noindent Moreover, as stated by Hansen  \textit{et al.} \cite{hansenetal2012} and discussed by Andersen \textit{et al.}  \cite{Andersenetal2003}, a single return offers only a weak signal about the current level of volatility. Thus, the implication is that GARCH models are poorly suited for situations where the volatility changes rapidly to a new level. The reason is that the GARCH model is slow at catching up, and it will take many periods for the conditional variance to reach its new level. 
Therefore, incorporating the realized measures in the GARCH model, one can alleviates this problem. In addition, with the advent of high-frequency data, several measures have been developing in the literature, such as the Realized Variance and Realized Kernel, among many others see Anderson and Bollerslev \cite{AndersonBollerslev1998}, Barndorff-Nielsen and Shephard \cite{BarndorNielsen2002} and Barndorff-Nielsen \textit{et al.} \cite{ BarndorNielsenetal2008}. 
All of these measures provide more information on the current level of volatility compare to the square of returns. This aspect makes the realized measures very important in modeling and forecasting future volatility. Therefore, by introducing the GARCH-X model, Engle \cite{Engle(2002)} incorporated the realized measures in the GARCH model. Hansen  \textit{et al.} 
\cite{hansenetal2012} introduced the Realized GARCH model by combining a GARCH structure for returns with an integrated model for realized measures of volatility. From this latter, several models have emerged.
For example, Hansen \textit{et al.} \cite{hansenetal2016} introduced the Realized EGARCH to capture the leverage effects. To highlight the property of long memory observed on the realized measure, Vander Elst \cite{HarryVanderElst2015}, introduced the FloGARCH model (\underline{F}ractiona\underline{l}ly integrated realized v\underline{o}latility \underline{GARCH}) and more recently, Sall \textit{et al.} \cite{sall2021} proposed the Realized HYGARCH model for modeling risk.\\

\noindent Nevertheless, as stated by Shi \textit{et al.} in \cite{shiandyang}, although these models bring some improvements, they have the same main weakness as the original GARCH model. This limit is the assumption that the conditional volatility has only one regime over the entire period. Furthermore, many studies show that structural change is common in financial datasets \cite{BeltrattiandMorana2006, EngleandRangel2008}. Diebold and Inoue \cite{DieboldandInoue2001} argue that the existence of structural change or stochastic regime-switching is not only related to the long memory, and they are generally easily confused. Therefore, a more appropriate volatility model would consider the long memory and the structural change simultaneously, (see, eg. \cite{BaillieandMorana2008, shiandyang}.) \\

\noindent Our aim of this work is to investigate an Adaptive Realized HYGARCH (A-Realized HYGARCH) model. 
This paper starts from the proposition that both long memory and structural breaks are likely to be present in the volatility processes of many economic and financial time series. It is designed for modeling the long memory of high-frequency financial time series with structural changes. This model incorporates the structure of the new HYGARCH model of Li et \textit{al.} in \cite{Lietal2014}, further considering the time-varying deterministic component in the flexible, functional form provided by Gallant in \cite{Gallant1984}.\\

\noindent This paper is organized as follows.  Section \ref{sec2} is dedicated to our proposed solution, the adaptive Realized HYGARCH model with structural breaks, while in section \ref{sec3}, we study the stability of the model. Section \ref{sec4} is reserved for the simulation of the experiment we did to evaluate our model. Section \ref{sec5} concludes this paper and gives some futur works.
\section{Adaptive Realized Hyperbolic GARCH}
 \label{sec2}
In this section, we present the Adaptive Realized HYGARCH model which takes into account structural changes and long memory of high-frequency data. Indeed, the incorporation of structural change in the long memory model is not a new idea in the literature. Baillie and Morana \cite{BaillieandMorana2008} presented the Adaptive Fractional Integrated GARCH (A-FIGARCH) model, which is designed for both long memory and regime change in financial time series. Shi and Yang \cite{shiandyang} developed also the adaptive Hyperbolic Exponential GARCH model. Following this methodology, we propose in this paper the Adaptive  Realized HYGARCH model which contains two components: a long stochastic memory part and a deterministic time varied function. The adaptive Realized HYGARCH $(p,q,d,k)$ model can be expressed as:\\
\begin{eqnarray}
\label{1}
 r_t&=&h_t^{1/2}z_t,\\
\label{2}
 \log h_t&=&\omega_t+\delta\left[1-\frac{1-\gamma(L)}{1-\beta(L)}(1-L)^d\right]\log x_t,\\
\label{3}
 \log x_t&=&\xi+ \phi \log h_t+\tau(z_t)+u_t
\end{eqnarray}
\noindent where 
$$\omega_t=\omega_0+\sum_{j=1}^k \left[ a_j\sin(2\pi jt/T)+b_j \cos(2\pi jt/T) \right],$$
\noindent $r_t$ is the return of the ime series, $x_t$ a realized measure of volatility, $\left(z_t\right)_t$ are  independently identically distributed (i.i.d) with  mean zero  and variance one and $\left(u_t\right)_t$ are also i.i.d  with mean zero and variance $\sigma_u^{2}$. Here $\left(z_t\right)_t$ and $\left(u_t\right)_t$  are mutually independent. We label Equation (\ref{1}) as return equation, Equation (\ref{2}) as the GARCH model and Equation (\ref{3}) as the measurement statement.\\
\noindent $L$ denotes the lag or backshift operator, $\beta\left(L\right)=\beta_1L+\beta_2L^2+\cdots+\beta_pL^p$ and $\gamma\left(L\right)=\gamma_1L+\gamma_2\L^2+\cdots+\gamma_qL^q$. The polynomial $\tau\left(z\right)=\tau_1 z+\tau_2\left(z^2-1\right)$ is called the leverage function and  facilitate a modeling of the dependence between return shocks and volatility shocks.\\
The main difference between the A-Realized HYGARCH model and the conventional Realized HYGARCH model is the inclusion of the time-varying intercept $(\omega_t)$. The A-Realized HYGARCH model can be reduced to the standard Realized HYGARCH model by setting $\omega_t=\omega_0(1-\beta(1))^{-1}$. \\

\noindent In the rest of this study, we consider the A-Realized HYGARCH$(1,d,1,k)$ model. The GARCH equation  is given by:
\begin{equation}
\label{eq1}
\log h_t=\omega_t+\delta\left[1-\frac{1-\gamma L}{1-\beta L}(1-L)^d\right]\log x_t.
\end{equation}
The fractional differencing operator has the following representation:
\begin{equation}
\label{eq2} 
 \left(1-L\right)^d=\displaystyle\sum_{k=0}^{\infty}\frac{\Gamma\left(d+1\right)\left(-L\right)^k}{\Gamma\left(k+1\right)\Gamma\left(d-k+1\right)}
\end{equation}
where $\Gamma$(.) denotes the Gamma function.
\section{Stability}
\label{sec3}
The stability of the model is one of the main property for any new model. Here, the stability refers to the behavior of the second moment of the model.
In this section,  we show that the second moment of the A-Realized HYGARCH model is asymptotically bounded under some conditions. The second moment of the model is calculated as 
$$E(r_t^2)=E(h_tz_t^2)=E(h_t).$$
We have 
$$\log h_t=\delta \log h_{1,t},$$
where 
\begin{eqnarray}
\label{4}
\nonumber \log h_{1,t}&=&\beta \log h_{1,t-1}-\frac{\beta}{\delta}\omega_{t-1}+\frac{1}{\delta}\omega_t+(\beta-\gamma+\pi_1)\log x_{t-1}\\
&\;&\quad\quad\quad\quad\quad\quad\quad\quad+\sum_{j=0}^\infty (\pi_{j+2}-\gamma \pi_{j+1})L^j \log x_{t-2}
\end{eqnarray} 
and 
\begin{eqnarray}
\label{5}
\nonumber \log h_t&=&\delta\beta \log h_{1,t-1}-\beta\omega_{t-1}+\omega_t+\delta(\beta-\gamma+\pi_1)\log x_{t-1}\\
&\;& \quad\quad\quad\quad\quad\quad\quad\quad\quad+\delta \sum_{j=0}^\infty (\pi_{j+2}-\gamma \pi_{j+1})L^j \log x_{t-2}.\end{eqnarray}.
\begin{Lem}
\label{lem1}
If $n_j$ and $m_j$ are the non negative numbers with $j\in$ \{1,2,$\cdots$,k \} such that 
$$\sum_{j=1}^k (n_j+m_j) \leq \min(1,\omega_0),$$
then
$$
0\leq \omega_t=\omega_0+\sum_{j=1}^k \left[ a_j\sin(2\pi jt/T)+b_j \cos(2\pi jt/T) \right] \leq 1+\omega_0+1=c_0. 
$$
\end{Lem}
\begin{proof}
For the proof of Lemma \ref{lem2}, one can refer to Charline et al. in \cite{charline}.
\end{proof}
\begin{Lem}
\label{lem2}
Let $\left(  V, \vert\vert.\vert\vert_\infty \right)$ be a normed space such that, 
$
V=\left\lbrace (y_t)_{t\in \mathbf{Z}} | \ \sup_{t\in \mathbf{Z}
}\mathbf{E}\vert y_t \vert \leq \infty \right\rbrace
$ 
and $L$ be a linear operator on $V$ defined by:
\begin{eqnarray*}
L&:& V\rightarrow V\\ 
&\;& y\rightarrow Ly=\left( Ly_t \right)_{t\in \mathbf{Z}}=\left( y_{t-1}\right)_{t\in \mathbf{Z}} 
\end{eqnarray*}
\noindent with $\vert\vert L \vert\vert_{\infty}=\sup_{t\in{\mathbf{Z}}} \mathbf{E}\vert y_t\vert$, then the delayed operator
satisfies
$$\vert\vert L^i\vert\vert_{\infty}=1\ \ \forall i\in \mathbf{N}.$$
\end{Lem}
\begin{proof}. For the proof of Lemma \ref{lem2}, one can refer to Charline et al. in \cite{charline}.
\end{proof}
\begin{Theo}
\label{tho1} 
The conditional variance $\log h_t$ of A-Realized HYGARCH model satisfies this follow inequality\\\\
$\mathbf{E}(\log h_t)\leq \vert \delta\beta \vert+f_0+\delta\phi \vert \beta-\gamma+\pi_1 \vert\mathbf{E}(\log h_{t-1})+\delta\phi \sum_{j=0}^{\infty}\vert \pi_{j+2}-\gamma\pi_{j+1} \vert \mathbf{E} (\log h_{t-2})$\\\\
where $f_0=c_0(1-\vert \beta\vert)+\xi\delta\vert \beta-\gamma+\pi_1 \vert+\xi\delta \sum_{j=0}^{\infty}\vert \pi_{j+2}-\gamma\pi_{j+1} \vert $
\end{Theo}
\begin{proof}
Using the equation $(\ref{5})$, the expectation of $\log h_t$ is given by:
\begin{eqnarray}
\label{6}
\mathbf{E}(\log h_t)=\delta\beta \mathbf{E}(\log h_{1,t-1})-\beta\mathbf{E}(\omega_{t-1})+\mathbf{E}(\omega_t)+\delta(\beta-\gamma+\pi_1)\mathbf{E}(\log x_{t-1})\nonumber\\  + \delta \sum_{j=0}^\infty (\pi_{j+2}-\gamma \pi_{j+1})L^j \mathbf{E}(\log x_{t-2}).
\end{eqnarray}
Since $\mathbf{E}(\log x_t)=\xi+\phi \mathbf{E}(\log h_t)$ then by  using Lemma $\ref{lem1}$ and Lemma $\ref{lem2}$, an upper bound of $(\ref{6})$ is calculated as follows:
$$
\left\{
    \begin{array}{l}
\displaystyle
\mathbf{E}(\omega_t) \leq c_0, \cr\cr
\delta\beta  \mathbf{E}(\log h_{1,t-1}) \leq \vert\delta\beta \vert \mathbf{E}(\log h_{1,t-1}),\cr\cr
\delta(\beta-\gamma+\pi_1)\mathbf{E}(\log x_t\omega_t)\leq  \delta\vert(\beta-\gamma+\pi_1)\vert(\xi+\phi \mathbf{E}(\log h_t)),\cr\cr
\delta \sum_{j=0}^\infty (\pi_{j+2}-\gamma \pi_{j+1})L^j\mathbf{E} (\log x_t)\leq \delta \sum_{j=0}^\infty \vert (\pi_{j+2}-\gamma \pi_{j+1})\vert (\xi+\phi \mathbf{E}(\log h_t)).\cr
\end{array}
    \right.
$$
By substituting the above results in (\ref{6}), we get
$$
\mathbf{E}(\log h_t)\leq \vert \delta\beta \vert+f_0+\delta\phi \vert \beta-\gamma+\pi_1 \vert\mathbf{E}(\log h_{t-1})+\delta\phi \sum_{j=0}^{\infty}\vert \pi_{j+2}-\gamma\pi_{j+1} \vert \mathbf{E} (\log h_{t-2}).
$$
\end{proof}
Now, consider the  A-Realized Hyperbolic GARCH process, we have
\begin{eqnarray}
\label{7}
\nonumber \mathbf{E}(\log h_t)&\leq& \vert \delta\beta \vert\mathbf{E}(\log h_{1,t-1})+f_0
+\delta\phi \vert \beta-\gamma+\pi_1 \vert\mathbf{E}(\log h_{t-1})\\
&\;& \quad\quad\quad\quad\quad\quad\quad\quad\quad+\delta\phi \sum_{j=0}^{\infty}\vert \pi_{j+2}-\gamma\pi_{j+1} \vert  \mathbf{E} (\log h_{t-2}),
\end{eqnarray}
and
\begin{eqnarray}
\label{8}
\nonumber \mathbf{E}(\log h_{1,t})&\leq& \vert \beta \vert\mathbf{E}(\log h_{1,t-1})+f_1+\phi \vert \beta-\gamma+\pi_1 \vert\mathbf{E}(\log h_{t-1})\\
&\;&\quad\quad\quad\quad\quad\quad +\phi \sum_{j=0}^{\infty}\vert \pi_{j+2}-\gamma\pi_{j+1} \vert  \mathbf{E} (\log h_{t-2}),
\end{eqnarray}
where $f_1=\frac{f_0}{\delta}$.  Note that inequalities $(\ref{7})$ and $(\ref{8})$ can be rewritten in matrix form as 
\begin{equation}
\label{9}
H_t \leq M+BH_{t-1},
\end{equation}
with some initial condition $H_{-1}$. Iterating inequality (\ref{9}), we get\\
 \begin{equation}
 \label{10}
H_t \leq M \sum_{i=0}^{t-1}B^i+B^tH_0=D_t.
\end{equation}
The matrices $H_t$, $M$ and $B$ are defined as follows:\\
$$H_t=\left[\begin{array}{c}
\mathbf{E}(\log h_t)\\
\mathbf{E}(\log h_{1,t})\\
\mathbf{E}( \log h_{t-1})
\end{array}\right]; \;
{M}=\left[\begin{array}{c}
f_0\\
f_1 \\
0
\end{array}\right]
 $$
and
$$ 
B=\left[\begin{array}{ccc}
\phi\delta\vert \delta(\beta-\gamma+\pi_1)\vert&\vert\beta \delta \vert &\phi\delta   \sum_{j=0}^\infty \vert(\pi_{j+2}-\gamma \pi_{j+1}))\vert\\\\
\phi\vert (\beta-\gamma+\pi_1)\vert&\vert\beta \vert &\phi \sum_{j=0}^\infty \vert (\pi_{j+2}-\gamma \pi_{j+1}))\vert\\\\
1    &      0        &    0
\end{array}\right].$$
\begin{Lem}
\label{lem3}
Let $\delta$, $\phi$, $\beta$ and $\gamma$ be the parameters of the A-Realized HYGARCH model. 
If
\begin{eqnarray*}
\left\lbrace\begin{array}{ccc}
\delta\phi\vert \beta-\gamma+\pi_1\vert+\vert\beta\vert +\phi\delta\sum_{j=0}^{\infty}\vert \pi_{j+2}-\gamma\pi_{j+1} \vert-1
&\leq& 0\cr\cr
\delta\phi\vert \beta-\gamma+\pi_1\vert+\vert\beta\vert & \leq & 2 ,
\end{array}\right.
\end{eqnarray*}
 then the spectral radius of $B$, $\rho(B)< 1$.
\end{Lem}
\begin{proof}
Let show that the spectrum $\Lambda(B)$ is not a empty set and its maximum eigenvalue is strictly less than one.\\
\begin{displaymath}
{B}=\left[\begin{array}{ccc}
\phi\delta\vert \delta(\beta-\gamma+\pi_1)\vert& \quad\quad \vert\beta \delta \vert &\quad\quad\phi\delta   \sum_{j=0}^\infty \vert(\pi_{j+2}-\gamma \pi_{j+1}))\vert\\\\
\phi\vert (\beta-\gamma+\pi_1)\vert&\quad\quad\vert\beta \vert &\quad\quad\phi \sum_{j=0}^\infty \vert (\pi_{j+2}-\gamma \pi_{j+1}))\vert\\\\
1    & \quad\quad     0        & \quad\quad   0
\end{array}\right]
\end{displaymath}
For sake of simplicity, let us rewrite the matrix $B$ as 
\begin{displaymath}{B}=\left[\begin{array}{ccc}
a&b &c\\\\
\frac{a}{\delta}&\frac{b}{\delta} &\frac{c}{\delta}\\\\
1     &      0        &    0
\end{array}\right]\end{displaymath}

\noindent The characteristic polynomial of $B$ is given by
$$
P_B(\lambda)=\lambda(-\lambda^2+(a+\frac{b}{\delta})\lambda+c).
$$
By solving the equation $P_B(\lambda)=0$, the eigen values of the  matrix $B$ are 
\begin{eqnarray*}
\left\lbrace\begin{array}{ccc}
\lambda_1&=&0\\
\lambda_2&=&\frac{1}{2}\left((a+\frac{b}{\delta})-\sqrt{(a+\frac{b}{\delta})^2+4c}\right)\\
\lambda_3&=&\frac{1}{2}\left((a+\frac{b}{\delta})+\sqrt{(a+\frac{b}{\delta})^2+4c}\right)
\end{array}\right.
\end{eqnarray*}
Since $max\{\lambda_1,\lambda_2, \lambda_3\}=\lambda_3$ that is, 
$$ 
\rho (B)= \frac{1}{2}\left((a+\frac{b}{\delta})+\sqrt{(a+\frac{b}{\delta})^2+4c}\right).
$$
The spectral radius of $B$ is less than one if and only if the following condition is satisfied 
\begin{equation}
\label{11}
\left\lbrace\begin{array}{cc}

a+\frac{b}{\delta}+c-1& \leq 0 \\
 a+\frac{b}{\delta} & \leq  2 

\end{array}\right.
\end{equation}
We just have to replace $a$, $b$ and $ c$ by thier expressions in $(\ref{11})$, where $a=\phi\delta\vert (\beta-\gamma+\pi_1)$,  $b=\vert\beta \delta \vert$ and $c=\phi\delta   \sum_{j=0}^\infty \vert(\pi_{j+2}-\gamma \pi_{j+1}))\vert$
Thus,   $(\ref{11})$ is rewritten as follows:
$$
\left\lbrace\begin{array}{cc}

\delta\phi\vert \beta-\gamma+\pi_1\vert+\vert\beta\vert +\phi\delta\sum_{j=0}^{\infty}\vert \pi_{j+2}-\gamma\pi_{j+1} \vert-1&\leq 0\\\\
\delta\phi\vert \beta-\gamma+\pi_1\vert+\vert\beta\vert & \leq  2 

\end{array}\right.
$$
\end{proof}
\begin{Theo}
\label{tho2}
Let $\delta$, $\phi$, $\beta$ and $\gamma$ be the parameters of the A-Realized HYGARCH model. If
$$
\left\lbrace\begin{array}{cc}

\delta\phi\vert \beta-\gamma+\pi_1\vert+\vert\beta\vert +\phi\delta\sum_{j=0}^{\infty}\vert \pi_{j+2}-\gamma\pi_{j+1} \vert-1&\leq 0\\\\
\delta\phi\vert \beta-\gamma+\pi_1\vert+\vert\beta\vert & \leq  2 ,

\end{array}\right.$$
then the process $\{r_t\}$ followings an A-Realized Hyperbolic GARCH model defined in relations $(\ref{1})$,$(\ref{2})$ and $(\ref{3})$ is asymptotically stable with finite variance.
\end{Theo}
\begin{proof}
From $(\ref{9})$ and $(\ref{10})$, we recall that
 \begin{equation*}
H_t \leq M \sum_{i=0}^{t-1}B^i+B^tH_0 \ \ \  t\leq 0.
\end{equation*}
According to the convergence matrix (see, Lancaster and Tismenetsky \cite{lancaster}),  the necessary and sufficient condition for the convergence of $D_t$ when $t\rightarrow \infty$ is $\rho(B)< 1$, by Lemma $\ref{lem3}$, suppose that the spectral radius is strictly less than one. Now we show that if $(I-B)$ exists, its inverse exists and $\sum_{i=0}^{t-1} B^i=(I-B)^{-1}$ as $\displaystyle\lim_{t\rightarrow \infty} B^tH_0=0$.\\
The eigenvalues of $(I-B)$ are $(1-\lambda(B))$ where $\lambda(B)$ are the eigenvalues of matrix $B$.\\ 
The set of eigenvalues of $(I-B)$ is not empty, hence matrix $(1-\lambda(B))$ is invertible.\\
Let 
$$ 
S_n=I+B+B^2+\cdots+B^{n-1}=\sum_{i=0}^{n-1} B^i$$
$$ BS_n=B+B^2+\cdots+B^n.
$$
Hence, 
$$
(I-B)S_n=I-B^n.$$
 By using the fact that $\displaystyle \lim_{n\rightarrow \infty} B^n=0$, we can prove that $\displaystyle\lim_{n\rightarrow \infty} (I-B^n)=I$. We get $(I-B) \displaystyle \lim_{n\rightarrow \infty}S_n=I$, that is  $\displaystyle \lim_{n\rightarrow \infty}S_n=(I-B)^{-1}$. More precisely  $ \displaystyle\lim_{n\rightarrow \infty} \sum_{i=0}^{n-1} B^i=(I-B)^{-1}$, as $\displaystyle\lim_{t\rightarrow \infty} B^t=0$, under Lemma $\ref{lem3}$, we conclude that $\displaystyle\lim_{t\rightarrow \infty} H_t \leq (I-B)^{-1} M .$
\end{proof}
\section{Estimation}
\label{sec4}
In this section, we report the Monte Carlo simulation evidence on the estimation of our Adaptive Realized Hyperbolic GARCH model for \textbf{Data Generating Processes (DGP)}. For all models used in this section, we assume that $z_t$ and $u_t$ follow respectively the student $\mathcal{T}$ distribution with $3$ degrees of freedom and the normal distribution $\mathcal{N}(0,\sigma_u)$. We consider $p=q=1$ and $\omega=0.1$, $ \gamma=0.1$, $\beta=0.4 $, $d=0.25,0.35;,0.45$ $, \delta =0.9$, $\epsilon =0$, $ \phi =1$, $\tau_1=-0.08$, $\tau_2=0.06$, and $ \sigma_u^2=0.4 $. The three values of the long memory  parameter $d$ are those proposed by \cite{shiandyang}, as \textit{low memory} ($d= 0.25$),\textit{ moderate memory} ($d=0.35$) and \textit{high memory} ($d=0.45$).\\

\noindent To obtain the \textbf{DGP} samples from Realized Hyperbolic GARCH with structural change, we fellow the \textbf{Step 1}, \textbf{Step 2} and  \textbf{Step 3} below. Notice that, \textbf{step 3} acts as the core part of this simulation study, and it must be repeated for each model and each replication. \textbf{Step 1} is also repeated for each replication while \textbf{Step 2} only needs to be performed once for each model. Following Shi and Yang in \cite{shiandyang}, we consider $500$ Monte Carlo replications.

 \begin{enumerate}
 \item{\textbf{Step 1}:} Set $z_t\backsim \mathcal{T}(0,1,\nu)$ and $u_t\backsim \mathcal{N}(0,\sigma_u)$. We get an i.i.d sample $\{z_t\}_{t=m}^T$ and $\{u_t\}_{t=m}^T$, where $m$ represents the number of extra burn in the data generated.
 \item{\textbf{Step 2}:} Choose appropriate designs for the intercept term in each model. 
 In this step, we consider three different designs:
 \begin{itemize}
 \item \textit{Design 1:}  $(m_1)$ assumes a constant intercept $\omega=\omega_t=0.1$, and corresponds to the standard experiment setting where no structural breaks are allowed in the conditional variance.
\item \textit{Design 2:} $(m_2)$ adopts the permanent break structure used by Baillie and Morana \cite{BaillieandMorana2008} and has one step change in the intercept. The intercept jumps from $0.1$ to $0.5$ without bouncing back in the future. 
Hence, 
\begin{displaymath} {\omega_t}=\left\lbrace\begin{array}{ccc}
0.1,&&t=1,...\frac{T}{2}\\
0.5,&&\quad t=\frac{T}{2}+1,...,T
\end{array}\right.\end{displaymath}

\item \textit{Design 3:} $(m_3)$ has two step changes. With the intercept jumping from $0.1$ to $0.5$ at the first break point and bouncing back to $0.3$ at the second break point. Hence, 
\begin{displaymath} {\omega_t}=\left\lbrace\begin{array}{ccc}
0.1,&&t=1,...\frac{T}{3}\\
0.5,&& \quad t=\frac{T}{3}+1,...,\frac{2T}{3}\\
0.3,&&\quad t=\frac{2T}{3}+1,...,T.
\end{array}\right.
\end{displaymath}
 \end{itemize}
\item{\textbf{Step 3}:} The sample $\{r_t\}_{t=1}^T$ and $\{x_t\}_{t=1}^T$ are obtained by using the specification Realized ARCH(3000).  
 \end{enumerate}

\noindent The log-likelihood function is applied on the models used in  this paper can be described as follows:
\begin{eqnarray}
     \label{eq15}
l\left(r,x;\theta_t\right)&=&\underbrace{-\sum_{t=1}^n \nonumber \left[A\left(\nu\right)+\log\left(\pi\left(\nu-2\right)\right)+0.5\log\left(h_t\right)+\dfrac{\nu+1}{2}\log\left(1+\frac{r_t^2}{h_t\left(\nu-2\right)}\right)\right]}_{l\left(r\mid x;\theta_t\right)}\\
&&\quad\quad\quad\quad\quad\quad\quad\quad\underbrace{-\frac{1}{2}\sum_{t=1}^n\left[\log  2\pi+\log\left(\theta_u^2\right)+\frac{u_t^2}{\theta_u^2}\right]}_{=l\left(x;\theta_t\right)},
\end{eqnarray}
 where 
 $$
 u_t=\log x_t-\epsilon-\phi \log h_t-\tau_1z_t-\tau_2(z_t^2-1)$$
  and  
  $$
  A(\nu)=\log(\Gamma(\frac{\nu}{2}))-\log(\Gamma(\frac{\nu+1}{2})).$$
 The parameter vector to be estimated is $\theta=(\omega',\gamma,\beta,d,\delta,\nu,\epsilon,\phi,\tau_1,\tau_2,\sigma_u^2,a',b' )$.\\
We maximize equation (\ref{eq15}) with the help to statistical packages in R software to estimate the vector $\theta$. \\
\begin{table}
\caption{Simulation results for estimation of A-Realized HYGARCH model without structural change.}
\label{tab1}
\begin{center}
\begin{tabular}{ ccccc }
\hline
 & A-Realized HYGARCH (1,d,1,0) &  &  &  \\ 
\hline
 d & Bias & RMSE & SE \\
\hline
 0.25 &0.0539&0.0932  & 0.0761&   \\
\hline
 0.35 &   0.0409  &   0.0868       &  0.0765  &  \\
\hline
  0.45 &  0.0363 & 0.0760  & 0.0667   &     \\
\hline
\end{tabular}\\
\end{center}
\end{table}\\
\textbf{Note 1}: The Table\ref{tab1} reports simulation results for the bias, root mean square error (RMSE) and standard error (SE) for estimation of the fractional differencing parameter d from simulations with sample size
T = 3000. All the results are based on 500 replications.\\

\begin{table}
\caption{Simulation results for estimation of A-Realized HYGARCH models without structural change.}
\label{tab2}
\begin{center}
\begin{tabular}{ ccccc }
\hline
 & A-Realized HYGARCH (1,d,1,k) &  &    \\ 
\hline
 d & Bias & RMSE & SE & \\
\hline
    &  k=1 &   &  &\\
\hline 
 0.25 & 0.0816   & 0.1184& 0.0858  & \\
  0.35 & 0.0239 &0.0795 & 0.0758& \\
0.45 & 0.0082 &  0.0793 & 0.0789& \\
\hline
   &  k=2 &   & &  \\
\hline   
  0.25 & 0.0563  & 0.0989 & 0.0813 & \\
 0.35 & 0.0090 & 0.0811 & 0.0806  & \\
  0.45 &  0.0059   & 0.0904 & 0.0903 & \\
\hline
 &  k=3 &   & &  \\
 \hline
  0.25 & 0.0302   &  0.0877 & 0.0823  &\\
  0.35 & 0.0028   & 0.0914  &  0.0914  &\\
  0.45 & -0.0276   &  0.1029  & 0.0992   & \\
\hline
 &  k=4 &   &   & \\
 \hline
  0.25 &  0.0390 & 0.0787 & 0.0684 & \\
  0.35  & -0.0132 & 0.0769  &  0.0758 & \\
  0.45 & -0.0411  &  0.1026  & 0.0940  &  \\
\hline           
\end{tabular}\\
\end{center}
\end{table}

\noindent \textbf{Note 2}: as for Table \ref{tab1}, the Table \ref{tab2} compute the estimated models used in Gallant \cite{Gallant1984}
's $k^{th}$ order flexible functional form, with $k = 1, 2, 3, 4$,  for the adaptive component.

\begin{table}
\caption{Simulation results for estimation of A-Realized HYGARCH$\left(1,d,1,k\right)$ models with various structural
change designs}
\label{tab3}
\begin{center}
\begin{tabular} {p{1.1cm}p{1.1cm}p{1.1cm}p{1.1cm}p{1.1cm}p{1.1cm}p{1.1cm}p{1.1cm}p{1.1cm}p{1.1cm}p{1.1cm}p{1.1cm}}
& & A-Realized HYGARCH(1,0.25,1,k) & & & A-Realized HYGARCH(1,0.35,1,k)& & &A-Realized HYGARCH(1,0.45,1,k) &  & &   \\
\hline 
 & & BIAS&RMSE & SE &BIAS & RMSE &SE& BIAS & RMSE & SE &\\
\hline
 k=0 &m2 &   0.0886 & 0.0959   & 0.0366   &   0.0105 & 0.0336 &  0.0319 & -0.058 &  0.0684 &  0.0362 &\\
  &m3 & 0.2353 & 0.2466  &  0.0737 &   0.1474    & 0.1590 & 0.0596 &  0.1127  &  0.1291   &  0.0630    &\\
\hline
 k=1 &m2 &0.1134  & 0.1203   &  0.0402&  0.0320 &  0.0487 & 0.0367  & 0.0329  &0.0452& 0.0419 &\\
  &m3 &0.1245  & 0.1366   &  0.0562 &  0.0753   & 0.0890  &  0.0474  &  0.0340   &  0.0605  & 0.0500  &\\   
 \hline
k=2 &m2 & 0.1244 & 0.1358   &  0.0544& 0.0307    & 0.0568  &  0.0478   & -0.046  & 0.0643 & 0.0443   &\\
  &m3 & 0.0895 & 0.1047    &    0.0542 &  0.0494  & 0.0719   & 0.0522   & 0.0108    & 0.0539  &  0.0528  &\\
\hline 
k=3 &m2 & 0.1481 &  0.1613   & 0.0640& 0.0666  & 0.0954    &   0.0682 &  -0.03   & 0.0647 &0.0569    &\\
  &m3 & 0.0566 & 0.0801 & 0.0567 &  0.0190  & 0.0671    & 0.0644     &  0.0011   & 0.0750    & 0.0750  &\\
\hline 
k=4&m2 &0.1416 &  0.1559  & 0.0653  &  0.0703& 0.1053  & 0.0784   &  -0.035 &  0.0599    & 0.0023 &\\
  &m3 & 0.0529&  0.0835   &   0.0646&  -0.004    &  0.0731 &   0.0730& -0.019& 0.0882 &  0.0860& \\   
\hline   

\end{tabular}
\end{center}
\end{table}
\noindent \textbf{Note 3}: The Table \ref{tab3} reports simulation results for the bias, root mean square error (rmse) and the standard error (se) for estimation of the fractional differencing parameter d from a sample size of $T = 3000$ observations. All the results are based on 500 replications. The simulations are for two different experiments of: a single break point $(m_2)$ and two break points $(m_3)$.\\

\noindent Table \ref{tab1} summarizes estimations results of the A-Realized HYGARCH models with $k=0$ equivalent to the ordinary A-Realized HYGARCH models for the Realized HYGARCH DGP with \textit{Design 1}. We remark that the estimated long memory parameter $d$ has a very small bias. This result is consistent with the three values assumed for $d$. The Realized HYGARCH$\left(1,d,1\right)$ DGP with $d=0.45$ has the lowest estimation bias.\\
Table \ref{tab2} summarizes estimation results for an A-Realized HYGARCH models with $k=(1,2,3,4)$. There is an important result obtained by comparing Table \ref{tab1} and Table \ref{tab2}. More than half of the model shows a reduction in bias after adopting an adaptive structure. As $d$ increases, the reduction in the degree of bias tends to increase. The estimated long memory parameter, obtained from the A-Realized HYGARCH and the Realized HYGARCH model estimation, has, in both, approximately the same degree of small sample RMSE. This result suggests that the intercept used (which follows a flexible function form with more than one pair of trigonometric components) can adjust for some uncertainties in the estimation of the long memory parameter $d$ \cite{BaillieandMorana2008}.\\
Table \ref{tab3} reports estimation results for estimates of A-Realized HYGARCH$\left(1,d,1,k \right)  $ models. From table \ref{tab3}, it can be seen that most A-Realized HYGARCH$\left(1,d,1,k \right)  $ models appear to have smaller estimation bias for the $m_3$ structural change design than the $m_2$ design. For the two cases, from the high persistence case $(d=0.45)$, the degree of bias in the estimates of $d$ is very small for both estimators. However, the bias is always smaller using the A-Realized HYGARCH model than the pure Realized HYGARCH model. Furthermore, the RMSE of the estimated of $d $ is generally lower from the A-Realized HYGARCH estimation compared to the pure Realized HYGARCH one. Finally, we can say that the A-Realized HYGARCH model is, generally, more perform than the standard Realized HYGARCH model. Indeed, the former is robust across the three values used in the designs contrary to the latter.  Furthermore, the improvement increases as the degree of persistence increases. \\

\noindent In general, the A-Realized HYGARCH model consistently outperforms the Realized HYGARCH model across different simulation designs with and without structural change. This fact suggests the usefulness of the A-Realized HYGARCH model in practice. The findings of this research are consistent with those from Baillie and Morana \cite{BaillieandMorana2008} 
\section{Conclusion and futur works}
In this article, we have developed the adaptive Realized HYGARCH process. It is much more flexible in modeling long-memory behavior and structural change often encountered in financial data. Under some assumptions, the model is shown to be stable. The quasi-maximum likelihood procedure is used to estimate the parameter of this model. Finite sample behaviors of this method were studied using Monte Carlo simulations. It indicates the A-Realized HYGARCH model outperforms the Realized HYGARCH model with and without structural change.\\
Since the results and the estimation methodology are encouraging, it will be interesting to examine the Adaptive Realized HYGARCH model's empirical application in financial data.\\

\end{document}